\documentclass{llncs}
\usepackage{amssymb}
\usepackage{amsmath,fancyhdr}
\usepackage{amsfonts}
\usepackage{epsfig}
\usepackage{enumerate}
\usepackage{verbatim}
\usepackage{graphics}
\usepackage{hyperref}
\usepackage{makeidx}
\usepackage{microtype}
\usepackage{subfig}
\usepackage{wrapfig}
\usepackage{pifont}

\usepackage{microtype}

\def\qed{ \ \vrule width.2cm height.2cm depth0cm\smallskip}

\def\cNP{\hbox{\rm \sffamily NP}}

\def\inst#1{$^{#1}$}

\def\wn{\mathrm{wn}}
\def\length{\mathrm{height}}

\newif\iflong

\longfalse

\iflong
\fi

\begin{document}

\title{C-planarity of Embedded Cyclic c-Graphs}

\author{Radoslav Fulek\inst{1}\thanks{The research leading to these results has received funding from the People Programme (Marie Curie Actions) of the European Union's Seventh Framework Programme (FP7/2007-2013) under REA grant agreement no [291734].}
\iflong\else\thanks{The omitted parts of the proof are in the appendix.}\fi}

\institute{
IST Austria, Am Campus 1, Klosterneuburg 3400, Austria\\
\email{radoslav.fulek@gmail.com}
}

\maketitle

\begin{abstract}
We show that c-planarity is solvable in quadratic time for 
flat clustered graphs with three clusters if the combinatorial embedding
of the underlying  graph is fixed. In simpler graph-theoretical terms our result can be viewed as follows. Given a graph $G$ with the vertex set  partitioned into three parts embedded on a 2-sphere, our algorithm decides if we can augment $G$ by adding edges 
without creating an edge-crossing so that in the resulting spherical graph the vertices of each part induce a  connected sub-graph.
We proceed by a reduction to the problem of testing the existence of  a perfect matching in planar
bipartite graphs.
We formulate our result in a slightly more general setting of cyclic clustered graphs,
i.e., the simple graph  obtained by contracting each cluster, where we disregard loops and multi-edges, is a cycle. 
\end{abstract}

%
%
%
%

\setcounter{page}{1}
\section{Introduction}

Testing planarity of graphs with additional constraints is a popular theme in
the area of graph visualizations.
One of most the prominent such planarity variants, c-planarity, raised in 1995 by Feng, Cohen and Eades~\cite{FCEa95,FCEb95} 
asks for a given planar graph $G$ equipped with a hierarchical structure on its vertex 
set, i.e., clusters, to decide if a planar embedding $G$ with the following property exists:
the vertices in each cluster are drawn inside a disc so that the discs
form a laminar set family corresponding to the given hierarchical structure
and the embedding has the least possible number of edge-crossings with the boundaries of the discs.
Shortly after, several groups of researchers tried to settle
the main open problem formulated by Feng et al. asking to decide its complexity
status, i.e., either provide
a polynomial/sub-exponential-time  algorithm for c-planarity or show its \cNP-hardness.
First, Biedl~\cite{B98} gave
 a polynomial-time algorithm for c-planarity with two clusters. A different approach
for two clusters was considered by Hong and Nagamochi~\cite{HN16}
and quite recently in~\cite{FKMP15}.
 The result also follows from a work by Gutwenger et al.~\cite{GJL+02}.
Beyond two clusters a polynomial time algorithm for c-planarity was obtained only in special cases,
e.g.,~\cite{BFPP08,GLS05,GJL+02,JJK+09,JKK+09}, and most recently in~\cite{BR14+,CBFK14+}. Cortese et al.~\cite{CDPP05} shows that c-planarity is solvable in polynomial
time if the underlying  graph is a cycle and the number of clusters
is at most three.

 In the present work we generalize the result of Cortese et al. to the class of all planar graphs
with a given combinatorial embedding. In a recent pre-print~\cite{F14+} we established 
a strengthening for trees, where we do not fix the embedding.
In the general case (including already the case of three clusters) of so-called flat clustered graphs a similar result was obtained
only in very limited cases. Specifically, either when every face of $G$ is incident
to at most five vertices~\cite{BF07,FKMP15}, or when there exist at most two vertices of a cluster incident to a single face~\cite{CBFK14+}.
We remark that the techniques of the previously mentioned papers do not give
a polynomial-time algorithm for the case of three clusters, and also do not seem to be adaptable
to this setting. Our result and the technique used to achieve it suggest that, for a fairly general class of clustered graphs, c-planarity could be tractable/solvable in  sub-exponential time at least with a fixed combinatorial embedding. 

 {\bf Notation.}
Let $G=(V,E)$ denote a connected planar graph possibly with multi-edges.
For  standard graph theoretical definitions such as path, cycle, walk etc.,
we refer reader to~\cite[Section 1]{D05}. 
A \emph{drawing} of $G$ is a representation of $G$ in the plane where every vertex
 in $V$ is represented by a unique point and every
edge $e=uv$ in $E$ is represented by a Jordan arc joining the two points that represent $u$ and $v$. 
We assume that in a drawing no edge passes through a vertex,
no two edges touch and every pair of edges cross in finitely many points.
An \emph{embedding} of $G$ is an  edge-crossing free drawing.
If it leads to no confusion, we do not distinguish between
a vertex or an edge and its representation in the drawing and we use the words ``vertex'' and ``edge'' in both
 contexts.
A  \emph{face} in an embedding is a connected component of the complement of the embedding 
of $G$ (as a topological space) in the plane.
 The \emph{facial walk} of $f$ is the closed walk in $G$ with a fixed orientation that we obtain by traversing the boundary of $f$ counter-clockwise.
In order to simplify the notation we sometimes denote the facial walk of a face $f$ by $f$. 
  A pair of consecutive edges $e$ and $e'$ in a facial walk $f$ creates a \emph{wedge} incident to $f$ at their common vertex.
  A vertex or an edge is \emph{incident} to a face $f$, if it appears on its facial walk.
The \emph{rotation} at a vertex is the counter-clockwise cyclic order of the end pieces of its incident edges
in a drawing of $G$.
An embedding of $G$ is up to an isotopy and the choice of an \emph{outer} (unbounded) face described by the rotations at its vertices. We call such a description of an embedding of $G$ a \emph{combinatorial embedding}. Remaining faces are \emph{inner faces}.
The \emph{interior} and \emph{exterior} of a cycle in an embedded graph is the bounded and unbounded, respectively, connected component
of its complement in the plane. 
Similarly, the \emph{interior} and \emph{exterior} of an inner face in an embedded graph is the bounded and unbounded, respectively, connected component
of the complement of its facial walk in the plane, and vice-versa for the outer face.
When talking  about interior/exterior or area of a cycle  
in a graph $G$ with a combinatorial embedding and a \emph{designated} outer face  we mean it with respect to an embedding in the isotopy class that $G$ defines.
For $V'\subseteq V$ we denote by $G[V']$ the sub-graph of $G$ induced by $V'$.



A \emph{flat clustered graph}, shortly  \emph{c-graph}, is a pair $(G,T)$, where $G=(V,E)$ is a graph and $T=\{V_0, \ldots, V_{c-1}\}$, $\biguplus_i V_i=V$, is a partition of the
vertex set into \emph{clusters}. See Figure~\ref{fig:treeEx} for an illustration.
A  c-graph $(G,T)$ is \emph{clustered planar} (or briefly \emph{c-planar}) if $G$ has an
 embedding in the plane such that (i)
for every $V_i\in T$ there is a topological disc $D(V_i)$, where $\mathrm{interior}(D(V_i))\cap \mathrm{interior} (D(V_j))=\emptyset$, if $i\not=j$,
 containing all the vertices of $V_i$ in its interior, and (ii)
 every edge of $G$ intersects the boundary of $D(V_i)$ at most once for every $D(V_i)$.
A c-graph  $(G,T)$ with a given combinatorial embedding of $G$ is \emph{c-planar} 
if additionally the embedding is combinatorially described as given.
 A \emph{clustered drawing and embedding} of a flat clustered graph $(G,T)$ is a drawing and embedding, respectively,
 of $G$ satisfying (i) and (ii).
In 1995
 Feng, Cohen and Eades~\cite{FCEa95,FCEb95} introduced the notion of clustered planarity for clustered graphs, shortly c-planarity, (using, a more general, hierarchical clustering)
as a natural generalization of graph planarity. (Under a different name
Lengauer~\cite{L89} studied a similar concept in 1989.)

\begin{wrapfigure}{r}{.5\textwidth}
  \centering
\centering
\subfloat[]{
\includegraphics[scale=0.4]{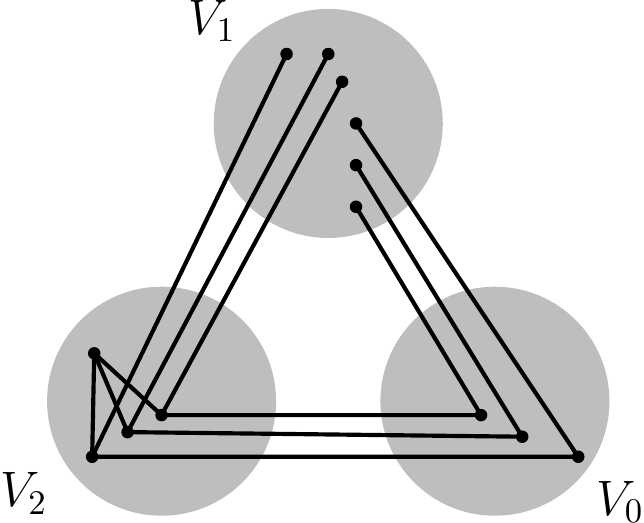}
    	}
\subfloat[]{
\includegraphics[scale=0.4]{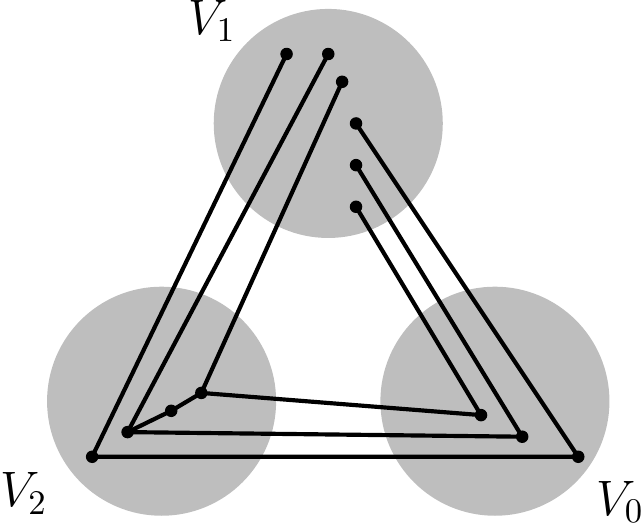}
		}
\caption{A c-graph that is not c-planar (left); and a c-planar c-graph (right).}
\label{fig:treeEx}
\end{wrapfigure}

By slightly abusing the notation for the rest of the paper $G$ denotes  a flat c-graph $(G,T)=(V_0 \uplus V_1 \uplus \ldots \uplus V_{c-1}, E)$ with $c$ clusters
$V_0,V_1,\ldots $ and $V_{c-1}$, and a given combinatorial embedding, and we assume that $G$ is \emph{cyclic}~\cite[Section 6]{FKMP15}. Thus, every $e=uv$ of $G$ is such that
$u\in V_i$ and $v\in V_j$ where $j-i \mod c \le 1$ and for every
$i$ there exists an edge in $G$ between $V_i$ and $V_{i+1 \mod c}$.
In the case of three clusters, the first condition is redundant.
If the second condition is violated, the problem was essentially solved for three clusters as discussed in
Section~\ref{sec:wnwn}.
We assume that $G$ is connected, since in the problem that we are studying, the connected components of $G$ can be treated separately. Indeed, without loss of generality  we  assume throughout the paper that in a clustered embedding of $G$ the clusters are unbounded wedges defined by pairs of rays emanating from the origin (see Figure~\ref{fig:wedges}) that is disjoint from all the edges (see Appendix). We call such a clustered drawing  a \emph{fan drawing}. \\

Thus, a connected component in a clustered embedding 
can be drawn so that it is disjoint from a ball $B$ centered at the origin of radius $\epsilon>0$
for any $\epsilon$. The rest of the graph is then embedded inductively inside $B$.
The aim of the present work is to prove the following.

\begin{theorem}
There exists a quadratic-time algorithm in $|V(G)|$ to test if a cyclic c-graph $(G,T)$ is c-planar.
\end{theorem}

{\bf Further research directions.}
%
We think that our technique should be extendable by means of Euler's formula to resolve the c-planarity in more general situations than the one treated in the present paper. In particular, we suspect that
the technique should yield a generalization of the characterization of strip planar clustered graphs~\cite[Section 5]{F14+}. 
That would allow us to work with graphs without a fixed embedding. We mention that the tractability in a special case 
of our problem known as cyclic level planarity, when the embedding is not fixed, follows from a recent work of Angelini et al.~\cite{angelini2015beyond}.

{\bf Organization.} In Section~\ref{sec:pre} we introduce concepts used in the proof of
our result. We give an outline of our approach in Section~\ref{sec:out}.
A more detailed description and a proof of correctness of our algorithm is in Section~\ref{sec:alg}.

\section{Preliminaries}
\label{sec:pre}

\subsection{Outline of the approach}
\label{sec:out}
By~\cite[Theorem 1]{FCEb95} deciding c-planarity of instances $G$ in which all $G[V_i]$'s are connected amounts to 
checking if an outer face of $G$ can be chosen so that every $V_i$ is embedded in the outer face 
of $G[V\setminus V_i]$. On the other hand, once we have a clustered embedding of $G$ we can augment $G$ by adding edges drawn inside clusters without creating an edge-crossing so that clusters become connected.
These observations suggest that c-planarity of $G$ could be viewed as a connectivity augmentation problem, for example as in~\cite{CBFK14+,FKMP15},
in which we want to decide if it is possible to make clusters connected  while maintaining the planarity of $G$.
One minor problem with this viewpoint is the fact that if $G$ is c-planar we do not allow
a cluster $V_i$ to induce a cycle such that clusters $V_j$ and $V_{j'}$, $i\not=j,j'$, are drawn on its opposite 
sides. However, this cannot happen  if $G$ is cyclic.
Following the above line of thought our algorithm tries to augment $G$ by subdividing its faces with
paths and edges. We proceed in two steps. In the first step, Section~\ref{sec:norm}, we either 
detect that $G$ is not c-planar or similarly as in~\cite{ADDF13} and~\cite{F14+} by 
turning clusters into independent sets and adding certain paths we normalize the instance. In the second step, Section~\ref{sec:const}, we decide if the normalized instance can be further augmented by  edges as desired.

In order to prove the correctness of the second step of the algorithm
we use the notion of the \emph{winding number} $\wn(W)\in \mathbb{Z}$ of a walk $W$ of $G$, as defined
in Section~\ref{sec:wnwn}. The parameter $\wn(W)$ says how many times and in which sense 
a walk $W$ of $G$ winds around the origin in a clustered drawing of $G$.
Thus, $G$ is not c-planar if there exists a face $f$ such that for its facial walk $|\wn(f)|>1$ or
if there exists at least two inner faces $f$ with $|\wn(f)|>0$.
However, it can be easily seen that this necessary condition of c-planarity is not sufficient
except when $G$ is a cycle~\cite{CDPP05}.
The  necessary condition allows us to reduce the c-planarity testing problem of a normalized instance to the problem
of finding a perfect matching in an auxiliary face-vertex incidence graph which is polynomially solvable.
The novelty of our work lies in the use of the winding number in the context of connectivity augmentation guided
by the flow and matching in the auxiliary face-vertex incidence graph \`a la~\cite{ADDF13} and~\cite{F14+}, respectively. 

We remark that the approach of~\cite{ADDF13} via a variant of upward embeddings
for directed graphs in our settings has several problems that seem quite hard to overcome,
the main one being the fact that the result of Bertolazzi et al.~\cite{BBLM94} does not extend, at least not in a natural way, to the drawings on the rolling cylinder, see e.g.,~Auer et al.\cite{Auer201536} for the definition of these drawings.
We are not aware of a polynomial-time algorithm for the corresponding  problem,
nor a corresponding \cNP-hardness result, and
find the corresponding algorithmic question interesting and related to our problem. 

\begin{figure}
  \centering
\centering
\subfloat[]{
\label{fig:wedges}
\includegraphics[scale=0.5]{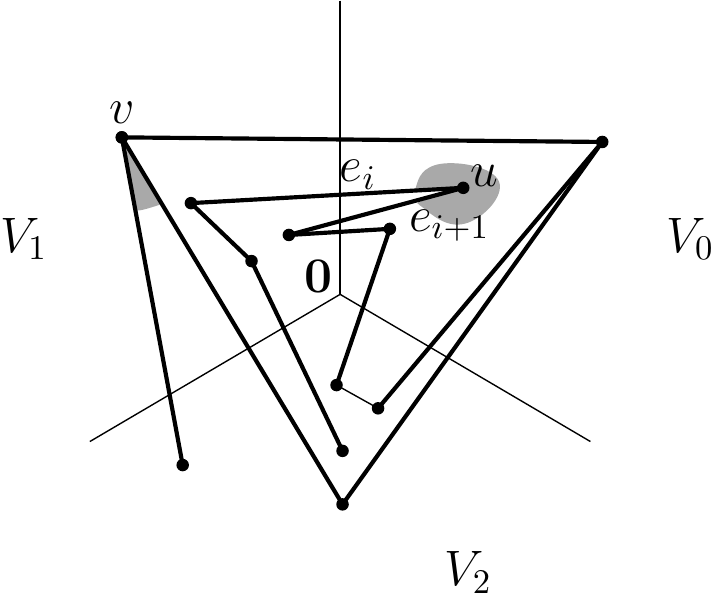}
    	}
    	\hspace{10px}
    	\subfloat[]{
    	\label{fig:semiSimple}
    	\includegraphics[scale=0.5]{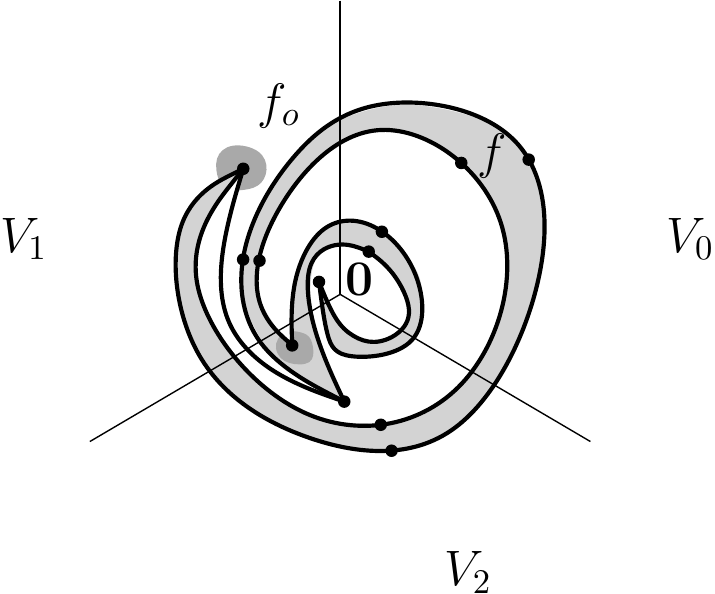}
    	}
\caption{(a) A clustered graph $G=(V_0 \uplus V_1 \uplus V_2,E)$ with clusters represented by wedges bounded by rays meeting at the origin. The highlighted wedge at $u$ is concave and at $v$ convex. (b) A semi-simple face $f$
and the outer face $f_o$ with an incident concave wedge.}
\end{figure}

\subsection{Winding number} 
\label{sec:wnwn1}
We define the winding number $\wn(W)$ of a closed oriented walk $W$ in a drawing disjoint from the origin of a graph $G$ (possibly with crossings). In what follows  facial walks are understood with the orientations as in
  an embedding of $G$ with the given rotations and a face $f_o$ being a designated outer face.
By viewing a closed walk $W$ in the drawing  as a continuous  function $w$ from the unit circle $S^1$ to $\mathbb{R}^2\setminus {\bf 0}$,
the winding number $\wn(W)\in \mathbb{Z}$ corresponds to the element of the fundamental group of  $S^1$~\cite[Chapter 1.1]{Hatch02}  represented by $\frac{w(x)}{||w(x)||_2}$.
Let $W_1$ and $W_2$ denote a pair of oriented closed walks meeting in a vertex $v$.
Let $W$ denote the closed oriented walk from $v$ to $v$ obtained by concatenating $W_1$ and $W_2$.
By the definition of $\wn$ we have $\wn(W)=\wn(W_1)+\wn(W_2)$.
Let $f_1$ and $f_2$, $f_o\not = f_1, f_2$, denote a pair of faces of $G$ whose walks intersect in a single walk.
Let $G'$ denote a graph we get from $G$ by deleting edges incident to both $f_1$ and $f_2$.
Let $f$ denote the new face thereby obtained. Since $f_1$ and $f_2$ intersect in a single walk, the boundary of $f$ is connected.
In the drawing of $G'$ inherited from the drawing of $G$ we have $\wn(f)=\wn(f_1)+\wn(f_2)$, since common edges of $f_1$ and $f_2$ are traversed
in opposite directions by $f_1$ and $f_2$.
A face  or a vertex is in the interior of a closed walk $W$ in $G$ if it is in the interior of 
a cycle induced by the edges of $W$ in an embedding of $G$ with the given rotations and  $f_o$ as the outer face.
The previous observation is easily generalized by a simple inductive argument 
as follows \\
$ {\bf (*)} \ \ \ \   \sum_f \wn(f)=\wn(W)$ \\
 where we sum over all  faces~$f$ of $G$ in the interior of the
closed walk $W$ in $G$. In particular, $\sum_f \wn(f)=\wn(f_o)$, where we sum over all 
 faces  $f\not=f_o$ of $G$.

\subsection{Labeling vertices} 
\label{sec:wnwn}
Let $\gamma:V \rightarrow \{0,1,\ldots c-1\}$ be a labeling of the vertices $V$ by integers
such that $\gamma(v) = i$ if $v\in V_i$.
 Let $W$ denote an oriented closed  walk in a clustered drawing of  $G$. We put $\length(W)=
\sum_{{v'u'}\in E(W)} g(\gamma(u')-\gamma(v'))$,
 where $g(0)=0, \ g(1)=g(1-c)=1$ and $g(c-1)=g(-1)=-1$.
We have the following.

\begin{lemma}
\label{lemma:wn}
For a walk $W$ in a fan drawing of $G$ we have $\wn(W)=\length(W)/c$.
\end{lemma}

\begin{proof}
The number of times
the walk $W$ crosses the ray between $V_i$ and $V_{i+1 \mod c}$ from right
to left w.r.t. to the direction of the ray is $\wn_i^+(W)=\sum_{{v'u'}} g(\gamma(u')-\gamma(v'))$, 
where we sum over the edges $v'u'$ in the walk $W$, where
 $v'\in V_i$ immediately precedes $u'\in V_{i+1 \mod c}$ in the walk.
 Similarly, we define \\  $\wn_i^-(W)=\sum_{{v'u'}} g(\gamma(u')-\gamma(v'))$, 
where we sum over the edges $v'u'$ in $W$, where
 $v'\in V_{i+1 \mod c}$ immediately precedes $u'\in V_{i}$ in the walk.
We have, $\wn(W) = \wn_i^+(W) + \wn_i^-(W)$
which in turn implies $c\cdot\wn(W) = \sum_{i} (\wn_i^+(W) + \wn_i^-(W))=\length(W) $.
\qed\end{proof}

By the previous lemma $\wn(W)$ is determined already by the c-graph $G$ and is the same in all  clustered drawings of $G$, and hence, putting $\wn(W):=\length(W)/c$, for a walk $W$ with a fixed orientation, allows us to speak about $\wn(W)$ without
referring to a particular drawing of $G$.
 Thus, $\wn(W)$ tells us the winding number of $W$ in any clustered 
drawing.
By Jordan-Sch\"onflies theorem $G$ the following holds. 

\begin{lemma}
\label{lem:2}
$G$ is not c-planar if there exists a face $f$ such that $|\wn(f)|>1$ or if there exists more than one inner face $f'$ with $|\wn(f')|=1$.
\end{lemma}
\begin{proof}
In a crossing free drawing  $|\wn(f)|\le 1$ for every face $f$.
If $|\wn(f')|=1$ the origin ${\bf 0}$ lies in the interior of $f'$ since
otherwise the facial walk is null-homotopic, i.e., homotopic to a constant map, in $\mathbb{R}^2\setminus {\bf 0}$ (contradiction). However, interiors of faces are disjoint.
\qed\end{proof}
 If $\wn(f)=0$ for all faces $f$,~\cite[Lemma 1.2]{F14+} extends easily to this case,
reducing the problem to the work of Angelini et al.~\cite{ADDF13}.
Thus, by Lemma~\ref{lem:2} and for the sake of simplicity of the presentation, throughout the paper we assume that there exists a pair  of faces $f_o,f_o'$, $\wn(f_o)=\wn(f_o')\not=0$  (by
~$(*)$ there cannot be just one such face) one of which, let's say $f_o$,
we designate as an \emph{outer face}.  The roles of $f_o$ and $f_o'$ are, in fact, interchangeable.
Also such a restriction is by no means crucial in our problem, and alternatively, it is always possible
to choose and subdivide the outer face in the normalized instance (defined later) by a path so that the restriction is satisfied.

Viewing a facial walk $f$ as a sequence of vertices and edges $w_0e_0w_1e_2\ldots e_mw_m$, where $e_{i-1}=w_{i-1}w_i$,
let $V_f$ be the set $\{w_0,\ldots, w_m\}$ of \emph{vertex occurrences} along~$f$.
We treat $V_f$ also as a multi-set of vertices, and thus, $\gamma$ is defined on its elements.
Let $\gamma_f:V_f \rightarrow \mathbb{N}$, for $f\not=f_o,f_o'$, be a labeling of the elements of $V_f$ by integers defined as follows.
We mark all the vertex occurrences  in $V_f$ as unprocessed.
We pick an arbitrary vertex occurrence  $v\in V_f$, set $\gamma_f(v):=\gamma(v)$
and mark $v$ as processed.
We repeatedly pick an unprocessed vertex occurrence  $u\in V_f$ that has its predecessor or successor $v$ along the boundary walk of $f$ in $V_f$  processed.
 We put  $\gamma_f(u):=\gamma_f(v)+g(\gamma(u)-\gamma(v))$.
 Intuitively, $\gamma_f$ records 
 the distance  in terms of ``winding around origin'' of vertex occurrences 
 along the boundary walk of $f$ from a single chosen vertex occurrence.
 Since $\wn(f)=0$ the function  $\gamma_f(u)$ is completely determined by
the choice of the  first occurrence of a vertex we processed. 
This choice is irrelevant for our use of $\gamma_f$ as we see later.
Also notice that $\gamma(v) = \gamma_f(v) \mod c$ for all vertices incident to $f$.

A normalized instance allows only the faces of the types defined next.
An element $v$ in $V_f$ is a \emph{local minimum} (\emph{maximum}) of a face $f$ if in the  facial walk  $f$ the value of $\gamma(v)$ is not bigger (not smaller)
with respect to the  relation $0<1<\ldots <c-1<0$  than the value of its successor and predecessor.
A walk $W$ in $G$ is \emph{(strictly) monotone with respect to $\gamma$} if the labels of the occurrences of vertices on $W$ form a (strictly) monotone sequence
with respect to the  relation $0<1<\ldots <c-1<0$ when ordered
in the correspondence with their appearance on  $W$.
The  face $f$ is \emph{simple} if $f$ has at most one local minimum. It follows 
that a simple face $f$ has also at most one local maximum.
The inner face   $f\not=f_o'$ is \emph{semi-simple}  if $f$ has exactly two local minima and maxima and these minima and maxima, respectively,  have the same $\gamma_f$ value.

\section{Algorithm}
 \label{sec:alg}

A cyclic c-graph $G$ is \emph{normalized } if 

\noindent
(i) $G$ is connected; \\
(ii) each cluster $V_i$ induces an independent set; and \\
(iii) each face of $G$ is simple or semi-simple, and $f_o$ and $f_o'$ are both simple. 

Suppose that (i)--(iii) are satisfied. 
By~(ii) we  put directions on all the edges in $G$ as follows.
Let $\overrightarrow{G}$ denote the directed c-graph obtained from $G$ by orienting every edge
$uv$ from the vertex with the smaller label $\min (\gamma(u), \gamma(v))$ to the vertex with the bigger label $\max (\gamma(u), \gamma(v))$ with respect to the relation $0<1<\ldots < c-1<0$.
A \emph{sink} and \emph{source} of $\overrightarrow{G}$ is
a vertex with no outgoing and incoming, respectively, edges.

Let $e$ denote an edge of $G$ not contained in a single 
cluster.
Given a clustered embedding $\mathcal{D}$ of $G$ let ${\bf p_{e}}:={\bf p_{e}}(\mathcal{D})$ denote the intersection point of $e$ with a ray separating a pair of clusters.
Let $e_0,\ldots,e_{k-1}$ be  the  edges
 incident to a sink or source $u$.
By Jordan curve theorem it is not hard to see that (i)--(iii) imply that a clustered embedding $\mathcal{D}$ of $G$ is ``combinatorially'' determined once we order the set 
  of intersection points
 ${\bf p_{e_0}}, \ldots , {\bf p_{e_{k-1}}}$ along
 rays separating clusters for every sink and sources $u$ in $G$. Moreover, the set of intersection points corresponding to a sink or source $u$ admits in
 an embedding only orders that are cyclic shifts of
 one another, since we have the rotations at vertices 
 of $G$ fixed.
 The wedge in $\mathcal{D}$ formed by a pair of edges $e_i$ and $e_{i+1}$ incident to a face $f$ at its local extreme $u$ is \emph{concave} (see Figure~\ref{fig:wedges} 
 for an illustration) if 
$u$ is a sink or source of $\overrightarrow{G}$ and
the line segment ${\bf p_{e_i}}{\bf p_{e_{i+1}}}$ contains all the other points ${\bf p_{e_j}}$ or
in other words the order of intersection points
corresponding to $u$ is ${\bf p_{e_{i+1}},p_{e_{i+2}}, \ldots, p_{e_{k-1}},p_{e_{0}}, \ldots , p_{e_i}}$.
A non-concave wedge is \emph{convex}.
Note that in $\mathcal{D}$ every sink or source is
incident to exactly one concave wedge that in
turn determines the order of intersection points.
Thus, combinatorially $\mathcal{D}$ is also determined
by a prescription of concave wedges at sink and sources.

Let  $S$ be the set of sinks and sources of $\overrightarrow{G}$. Let $F$ denote the union of the set of 
semi-simple faces of $G$ with a subset of $\{f_o,f_o'\}$ containing faces incident to a sink and a source.
We construct a planar bipartite graph $I=(S\cup F, E(I))$ with parts $S$ and $F$,
 where $s\in S$ and $f\in F$
is joined by an edge if $s$ is incident to $f$. 
Given that (i)--(iii) are satisfied, the existence of a perfect matching $M$ in $I$ is a necessary condition for $G$ being c-planar. Indeed, as we just said,
in a clustered embedding, each source or sink has exactly one of its wedges concave.
On the other hand, by Jordan curve theorem  it  can be easily checked that in the clustered embedding  \\
{\bf (A)} every semi-simple face is incident to exactly one concave wedge \\
{\bf (B)} faces $f_o$ and $f_o'$ are incident to one concave wedge if they are incident to a sink and source, and \\ {\bf (C)} all the other faces are not incident to any concave wedges at the minimum and maximum. \\
This is fairly easy to see if $G$ is vertex two-connected, see Figure~\ref{fig:semiSimple} for an illustration.
The cycle $C$ \emph{corresponding to a closed walk} is obtained
by traversing the walk and introducing a new vertex for each
vertex occurrence in the walk.
For a face $f$ incident to cut-vertices, {\bf (A)--(C)} follows by considering the cycle corresponding to the facial walk of $f$ (treated as a face) embedded in a close vicinity of the boundary of $f$. 
 Thus, a desired matching $M$ is obtained by matching each source or sink with
the face incident to its concave wedge.

We show in Section~\ref{sec:const}  that if $M$ exists $G$  is c-planar by augmenting 
$G$ with edges as described in Section~\ref{sec:out}.
Testing the existence, but even counting perfect matchings in a planar bipartite graph can be carried out in a polynomial time~\cite[Section 8]{L09}.

The running time of our algorithm is $O(|V|^2)$ since finding the perfect matching 
can be done in $O(|V|^2)$  time, due to $|E(I)|=O(|V|)$, and the pre-processing step including the construction of $I$ and the normalization will be easily seen to have this time complexity. Also computing the winding number for all the faces can be performed in a linear time by Lemma~\ref{lemma:wn}.
First, we explain and prove the correctness of the algorithm for
 instances satisfying~(i)--(iii). In Section~\ref{sec:norm},
 we show a polynomial-time reduction of the general case  to instances  satisfying~(i)--(iii). We often use Jordan--Sch\"onflies theorem without
 explicitly mentioning~it.

\iflong\else
\fi
\subsection{Constructing a clustered embedding}
 \label{sec:const}
 
\begin{wrapfigure}{r}{.5\textwidth}
\centering
\includegraphics[scale=0.65]{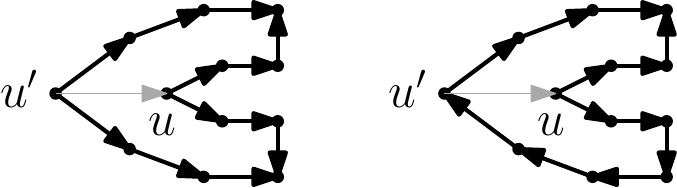}
\caption{Subdividing a semi-simple face   (left). Subdividing a simple face $f_o'$ (right).}

\label{fig:figfig}
\end{wrapfigure}

 Given a normalized instance $G$ and a matching $M$ between sources and sinks in $S$, and  faces in $F$ of $G$ we construct a  clustered embedding of $G$ as follows. Recall that we assume that $G$ does not have a face 
 $f$ with $|\wn(f)|>0$ besides $f_o$ and $f_o'$.
  We start with $\overrightarrow{G}$ defined
above and add edges to it thereby eliminating all the sinks and sources, see Figure~\ref{fig:figfig}.
Let $u\in S$ be a source matched in $M$ with $f$. If $f$ is a semi-simple inner face
let $u'$ denote another local minimum  incident to $f$. We add to  $\overrightarrow{G}$
  an edge $\overrightarrow{u'u}$ embedded in the interior of $f$. 
  If $f=f_o$ or $f=f_o'$ we join $u$ by $\overrightarrow{u'u}$ with the vertex in the same cluster $u'$
so that we subdivide $f$ into two simple faces $f'$ and $f''$ such that
$\wn(f')=0$ and $\wn(f'')=\wn(f)$. If $f=f_o$  face $f''$ is the new outer face. By Lemma~\ref{lemma:wn}, such a vertex $u'$ exists and it is unique.

%
%

We proceed with $u\in S$ that are sinks analogously thereby eliminating all the sinks and source in the resulting graph $\overrightarrow{G'}$, where by $G'$ we denote its underlying undirected graph.
By Lemma~\ref{lemma:wn}, there still exists exactly one inner face $f_o'$ with a non-zero winding number in the resulting graph $G'$.

\begin{lemma}
\label{lemma:keyFact0}
$G'$ has exactly one inner face $f_o'$ such that $|\wn(f_o')|=1$.
\end{lemma}

Since $\gamma(v) = \gamma_f(v) \mod c$ for every face $f\not=f_o,f_o'$ and $v$ incident to $f$, 
every edge we added joins a pair of vertices in the same cluster.

    \begin{lemma}
\label{lemma:keyFact1}
The induced sub-graph $G'[V_i]$ of (undirected) $G'$ does not contain a cycle for $i=0,1,\ldots, c-1$.
\end{lemma}
\begin{proof}
For the sake of contradiction suppose that a cycle $C$ is contained in $G'[V_{j'}]$.
Let us choose $C$ such that the area of its interior is minimized.
Since $G[V_{j'}]$ is an independent set all the edges of $C$ are newly added.
Thus, by looking at the rotation of an arbitrary vertex $v'$ of $C$ we see that $v'$ is incident to a vertex 
 $v$ from  $V_j$, $j\not=j'$,  in the interior of $C$. Indeed, no two edges of $C$ subdivide the same face of $G$.

Using the fact that $\overrightarrow{G'}$ does not contain any source or sink, we show that 
a vertex $w$ in the interior of $C$ belongs to an oriented cycle $C'$ (by chance also directed in $\overrightarrow{G'}$), whose interior is contained in the interior of $C$   
such that $\wn(C')>0$.
The cycle $C'$ is obtained by following a directed path in $\overrightarrow{G'}$ (from which it inherits its orientation)
 passing through  $v$.
Either both ends of the path meet each other, they both meet $C$, or the path meet itself in the interior. 
In the first two cases we can take $w:=v$ in the last case it can happen that the directed path gives rise to a cycle $C'$ not containing $v$. However, $C'$ is not induced by a single cluster by the choice of $C$, and thus, $\wn(C')>0$ by Lemma~\ref{lemma:wn} and $C'$ contains a vertex $w$ from $V_j$.
 Let $F'$ denote the set of faces
in the interior of $C$ and not in the interior of $C'$. 
In all cases it can be seen   by Lemma~\ref{lemma:wn} that $\wn(C')>0$.

Indeed, as we proved in  the proof of Lemma~\ref{lemma:wn} 
$\wn(C') = \wn_j^+(C') + \wn_j^-(C')$.
Since $C'$ follows a directed path and is not induced by a single cluster we have $\wn_j^+(W)>0$ and $\wn_j^-(W)=0$.
Hence, $\wn(C') = \wn_j^+(C') + \wn_j^-(C')>0$.

 By~(*) it follows that $C'$ contains the unique
inner face with a non-zero winding number in its interior.
Then Lemma~\ref{lemma:keyFact0} with~(*) yields the following  contradiction
$$0=\wn(C)=\wn(C')+\sum_{f\in F'} \wn(f)=\wn(C')\not=0$$
\qed\end{proof}

Let $E'\subseteq \bigcup_i{V_i \choose 2}\setminus E(G')$ such that each edge in $E'$ can be added to 
the embedding of $G'$ without creating a crossing or increasing the number of inner faces with a non-zero winding number.
We do not put any direction on the edges in $E'$.
Since every inner face $\not=f_o'$ in $G'$ is simple, and its outer face and the face $f_o'$ are not adjacent
to a source or sink, all the edges in $E'$ can be 
introduced simultaneously without creating a crossing. 
In particular, no edge of $E'$ subdivides $f_o'$ or the outer face.
Let $E''$ denote a maximal subset of $E'$
that does not introduce a cycle in $(G'\cup E'')[V_i]$ for every $i=0,1,\ldots, c-1$ (see Figure~\ref{fig:simpleFace}), where $G' \cup E'' = (V(G'), E(G') \cup E'')$.
By Lemma~\ref{lemma:keyFact1}, $E''$ is well-defined.

\begin{figure}
\centering
\includegraphics[scale=0.65]{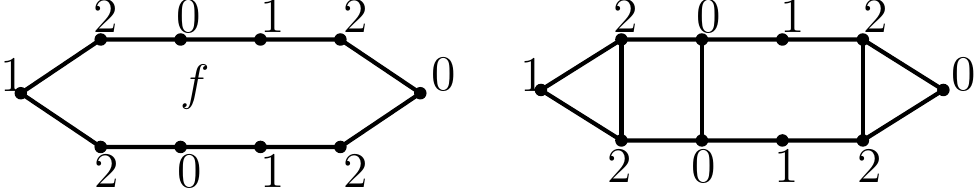}
\caption{A simple face $f$ of $G'$ (left). The face $f$ subdivided with edges of $E''$ (right). Labels at vertices are their $\gamma$ values (or indices of their clusters).}
\label{fig:simpleFace}
\end{figure}

\begin{lemma}
\label{lemma:keyFact2}
  $(G'\cup E'')[V_i]$  is a tree for $i=0,1,\ldots, c-1$.
  \end{lemma}
\begin{proof}
Suppose for the sake of contradiction that $(G'\cup E'')[V_i]$ for some $i$ is not a tree, and thus, it is just a  forest with more than one connected component.
It follows that either (1) there exists a cycle in $(G'\cup E'')[V\setminus V_i ]$
containing a vertex $v$ of $V_i$ in its interior
or (2) a pair of vertices of $V_i$ in different
connected components of $(G'\cup E'')[V_i]$  are incident to the same face of $(G'\cup E'')$.  
The claim (1) or (2) implies that there exists a cycle
$C$ in $(G'\cup E')[V\setminus V_i]$ containing a vertex $w$ of $V_i$ in its interior.
Similarly as in the proof of Lemma~\ref{lemma:keyFact1}, by following
a directed path through $v$ we obtain 
an oriented cycle $C'$ (this time not necessarily directed) in $G$,  whose interior is contained in the interior of $C$ with $\wn(C')>0$ yielding a contradiction.

Indeed, as we proved in  the proof of Lemma~\ref{lemma:wn} 
$\wn(C') = \wn_i^+(C') + \wn_i^-(C')$.
Since $C'$ is not induced by a single cluster and follows in the interior of $C$ a directed path, and $C$ does not have any vertex in $V_i$ we have
 $\wn_i^+(W)>0$ and $\wn_i^-(W)=0$.
Hence, $\wn(C') = \wn_i^+(C') + \wn_i^-(C')>0$. 
\qed\end{proof}

By Lemma~\ref{lemma:keyFact2}, every $F_i$ is a tree. Taking a close neighborhood of each such $F_i$ as a disc representing
the cluster $V_i$ we obtain a desired clustered embedding of $(G'\cup E'')$. In the obtained embedding we just delete edges not belonging to $G$ and that concludes the proof of the correctness of our algorithm. \\

\subsection{Normalization}
 \label{sec:norm}

 In the present section we normalize the instance so that~(i)-(iii) are satisfied.
We argued the connectedness in Introduction, and hence, (i) is taken care of.
\iflong
\else
To achieve~(ii) is fairly standard by contracting components induced by clusters to vertices. \fi
Thus, it remains to satisfy (iii).

We want to sub-divide a non-simple face $f$  into 
a pair of faces one of which is
semi-simple by a monotone path $P'$ w.r.t. $\gamma$.
Let $uPv$ denote an oriented monotone  sub-walk of $f$ w.r.t. $\gamma$ joining a local minimum $u$ and maximum $v$ of $f$ minimizing  $|\length(P)|$.
Let $vQv'$ denote the oriented  monotone walk with $|\length(P)|=|\length(Q)|$ immediately following $P$ on the facial walk of $f$, and let $u'Q'u$ be such walk  immediately preceding $P$ on the facial walk of $f$. Note that $Q$ and $Q'$ exists due to the minimality of $P$ and that we have
$\length(Q)=\length(Q') = - \length(P)$.
Similarly as in~\cite{F14+} we subdivide $f$ into two faces $f'$ and $f''$
by a strictly monotone path $v'P'u'$  w.r.t. $\gamma$. Hence,  $\length(P) = \length(P')$.
We have $\length(Q)=\length(Q') = - \length(P)= -\length(P')$.
Thus, by Lemma~\ref{lemma:wn} if $f$ with $\wn(f)\not=0$ is semi-simple we obtain a simple face $f'$ with 
$\wn(f')\not=0$ and a semi-simple face $f''$ with $\wn(f'')=0$ as desired.
Indeed, $\wn(f'') = \length(P') +\length(Q') + \length(P)+\length(Q)=0$
and $c\cdot\wn(f) = \length(v'P''u') + \length(Q')+\length(P) + \length(Q) =\length(v'P''u') - \length(P')= c\cdot\wn(f')$. It remains to show the following lemma, since both $f'$ and $f''$ 
are incident to less local minima and maxima than $f$ if $f$ is not semi-simple.
Hence, after $O(|V|)$ facial subdivisions we obtain a desired instance, since $|E(I)|=O(|V|)$.

\begin{lemma}
\label{lemma:norm}
If the c-graph $G$ is c-planar then by subdividing $f$ of $G$ by $P'$ into
a pair of faces $f'$ and $f''$, where $f''$ is semi-simple we obtain a c-planar c-graph. Moreover, 
$\wn(f')=\wn(f)$ and $\wn(f'')=0$.
\end{lemma}
\begin{proof}
The second statement is proved above.
Hence, we deal just with the first one.
Let $e_{u}$ and $e_{u}'$ denote the first edge on $P$ and the last edge on $Q'$, respectively. Let $e_{v}$ and $e_v'$
denote the last edge on $P$ and the first edge on $Q$, respectively.
Let $e_{v'}$ and $e_{u'}$ denote the last edge on $Q$ and the first edge on $Q'$.
Let  ${\bf p_{u}=p_{e_{u}}}, $  and ${\bf p_{v'}=p_{e_{v'}}}$
 denote the intersection of the edges $e_u,$ and $e_{v'}$, respectively,  with a ray separating a pair of clusters.
Let $\omega_u$ and $\omega_v$ denote the wedge between $e_u,e_u'$ and $e_v,e_v'$, respectively, in $f$.

We presently show that subdividing $f$ with $P'$  preserves c-planarity, since a clustered
embedding without $P'$ can be deformed so that $P'$ can be added to a clustered planar embedding without creating a crossing, while keeping the embedding clustered. 
This is not hard to see if, let's say $\omega_v$, is convex and the line segment ${\bf p_up_{v'}}$ is not crossed by an edge. Since $\omega_v$ is
convex, the relative interior of ${\bf p_up_{v'}}$ is contained in the interior of $f$. Note that $u'Q'PQv'$ is a sub-walk of $f$ since $f$ is not simple. We draw a curve $C$ joining $u'$ with $v'$ following the walk $u'Q'PQv'$ in its small neighborhood in the interior $f$; we cut $C$ at its (two) intersection points with ${\bf p_up_{v'}}$ and reconnected the severed ends on both sides by a curve following ${\bf p_up_{v'}}$ in its small neighborhood thereby obtaining a closed curve, and a curve $C'$ joining  $v'$ and $u'$. Finally, $C'$ can be subdivided by vertices thereby 
yielding a desired embedding of $G\cup P'$.
Otherwise, if $\omega_v$ is concave
or ${\bf p_up_{v'}}$ is crossed by an edge of $G$ we need to deform the clustered
embedding of $G$ so that this is not longer the case.

 \begin{figure}[h]
  \centering
\centering\textbf{•}
{
\includegraphics[scale=0.6]{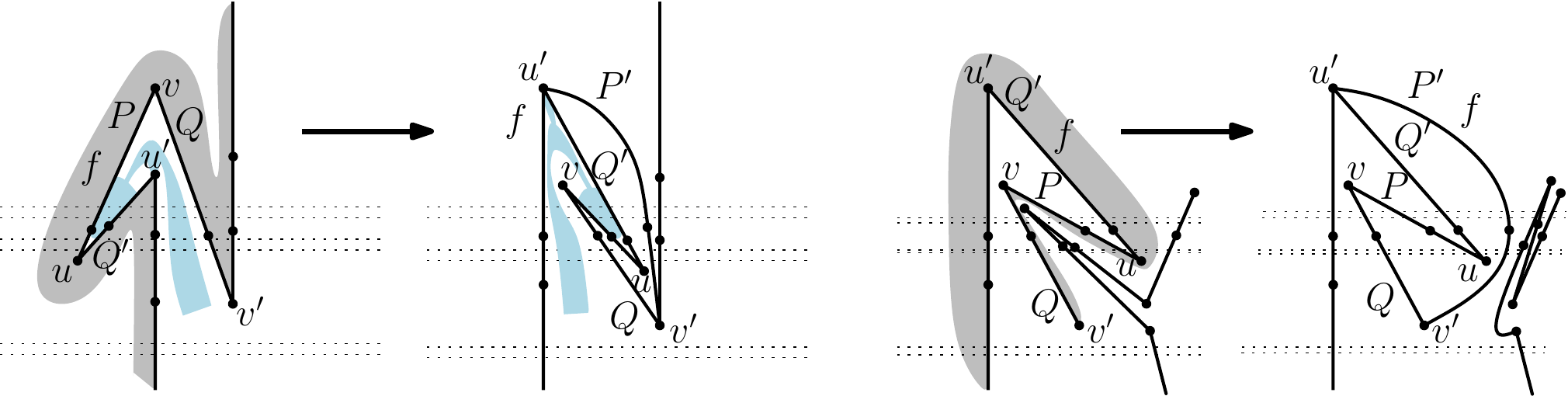}\textbf{•}
    	}

\caption{A pair of deformations of the clustered embedding of $G$ so that $f$ can be subdivided by $P'$. For the sake of clarity clusters are drawn as horiznotal strips rather than wedges.}
 \label{fig:stork}
\end{figure}

By a \emph{spur} with the \emph{tip} $u$ we understand a closed curve  obtained  as  a concatenation of a line segment contained in a ray separating clusters
and a curve contained in the boundary of $f$ passing through exactly one extreme $u$
of $f$ such that the curve is longest possible. The \emph{length} is the spur
is one plus the number of its crossings with rays separating clusters divided by two.
 If $\omega_u$ is concave, the vertex $u$ is a tip of a spur whose length is the distance of $u$ to a closest other extreme along the face. Note that both $P$ and $Q'$ must be paths in this case.
The rough idea in the omitted part of the proof is that shortest spurs have room around them to be deformed while maintaining the embedding
clustered such that $P'$ can be added.
Spurs  are deformed  as illustrated in Fig.~\ref{fig:stork}. (see Appendix for the rest of the proof)
\qed\end{proof}
%



{\bf Acknowledgment} I would like to thank Jan Kyn\v{c}l and D\"om\"ot\"or P\'alv\"olgyi for many comments and suggestions that helped to improve the presentation of the result.

\bibliographystyle{plain}

\bibliography{bib}

\newpage

\section*{Appendix}

\subsection*{Fan drawings}
\label{sec:fan}
We show that the clusters can be drawn as regions, each bounded by a pair of rays emanating from the origin.
Suppose that $G=(V_0\uplus \ldots \uplus V_{c-1},E) $ is given by a clustered embedding
living in the $xy$ plane of $\mathbb{R}^3$.
We assume that boundries of discs representing clusters do not touch.
Consider a stereographic projection from the north pole of a two-dimensional sphere $S$ 
sitting at the origin of $\mathbb{R}^3$.
Let $D$ be a stereographical pre-image of the embedding of $G$ on $S$.
Let $S'$ denote the union of $G$ (as a topological space) with the boundaries of the clusters in $D$.
Let $R_n$ and $R_s$ be a connected component of the complement of $S'$ in $S$, respectively, containing the north pole and south pole.
If necessary, we apply an isotopy  to $D$ (a continuous deformation keeping $D$ to be a clustered embedding all the time)  so that in the resulting embedding $D$ of $G$ on $S$ every boundary of a cluster intersects (in fact touches) the closure of $R_n$ and  the closure of $R_s$. 

We show that a desired isotopy exists. We contract every cluster to a point thereby
treating clusters as vertices in an embedding $D'$ of a cycle $C$ of length $c$ having multi-edges. 
Formally, this can be viewed as a quotient $S/\sim$, where $x\sim y$ iff
$x$ and $y$ belong to the same cluster.
In $D'$ there must be a pair of faces $f$ and $f'$ whose facial walk is $C$ since any cycle in the corresponding multi-graph is obtained as a symmetric difference of facial walks. Apply an isotopy to $D'$ such that $f$ contains
the north pole in its interior and $f'$ contains the south pole in its interior. Finally, we decontract clusters in the end. The above procedure can be easily turned into an isotopy of $D$. 

By projecting the resulting spherical embedding back to the plan we can also assume that we have a clustered embedding of $G$
such that clusters are represented by small discs of diameter $\epsilon>0$ each drawn in a close vicinity
of a different vertex of a regular convex $c$-gon with the center at the origin, and the edges
between clusters $V_i$ and $V_{i+1 \mod c}$, for every $i$, are closely following the edge
of the $c$-gon between the corresponding pair of vertices. 
The desired rays bounding clusters are those from the origin orthogonal to the sides of the $c$-gon. \\

\section*{Normalization}

\begin{proof}[the omitted part of the proof from Section~\ref{sec:norm}]
 \begin{figure}[h]
  \centering
\centering
{
\includegraphics[scale=0.65]{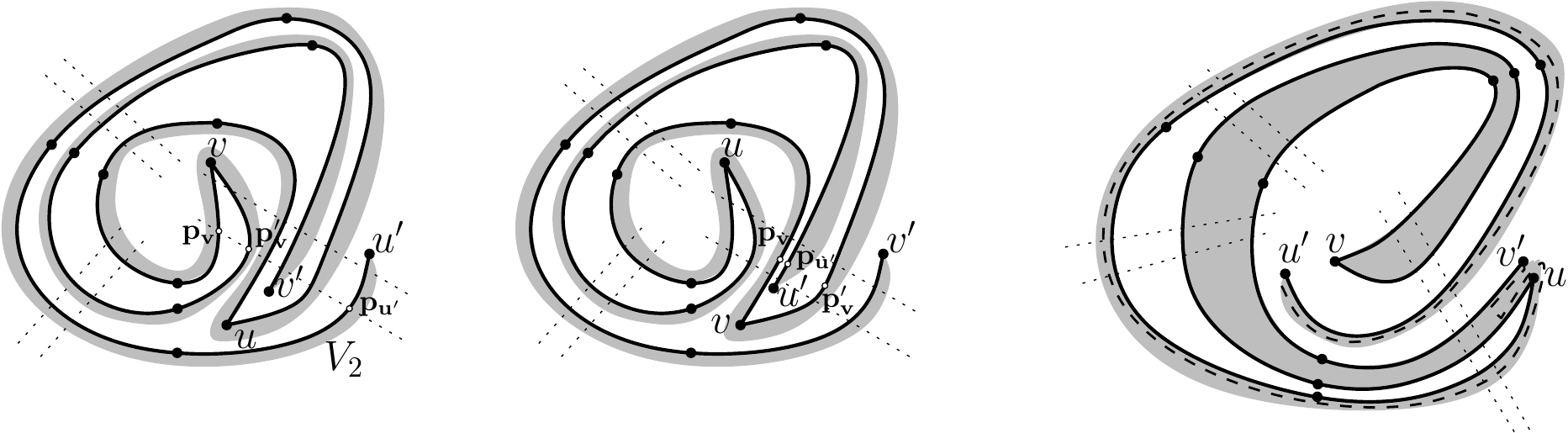}
    	}

\caption{Deformation in the case when both $u$ and $v$ have concave wedges incident to $f$ that is indicated by grey.
The dashed curve represents the path $P'$ subdividing $f$.
On the left,  point ${\bf p_v'}\in{\bf p_{v}p_{u'}}$. In the middle, 
   point ${\bf p_v'}\not\in{\bf p_{v}p_{u'}}$. On the right, the corresponding
 deformation.}
\label{fig:deform1}
\end{figure}

 First, we suppose that $\omega_u$ is concave.
 W.l.o.g. we assume that ${\bf p_{v}'}\not\in{\bf p_{v}p_{u'}}$. This holds when $\omega_v$ is convex,
 Figure~\ref{fig:deform2}. Otherwise,
we exchange the roles of $u$ and $v$, see Figure~\ref{fig:deform1}. Combinatorially, there are two
cases depending on whether $v$ is concave, but we treat them
simultaneously.
We isolate a part of the embedding of $G$ inside a spur represented
by a topological disc $D$. In order to get a desired deformed clustered embedding of $G$ we define a homeomorphism from $D$ that we use to redraw the corresponding part of $G$  thereby disconnecting some edges
that are reconnected in the end.
 Let $D_0$ denote the topological disc bounded by the closed curve
obtained by concatenating the line segment ${\bf p_{v}}{\bf p_{u'}}$ with the parts
of $P$ and $Q'$ connecting endpoints of ${\bf p_{v}}{\bf p_{u'}}$  with $u$. 
We assume that $v'\not\in D_0$ which holds automatically when $\omega_v$ is concave due to ${\bf p_{v}'}\not\in{\bf p_{v}p_{u'}}$.

 \begin{figure}[h]
  \centering
\centering
\subfloat[]{
\label{fig:deform3}
\includegraphics[scale=0.65]{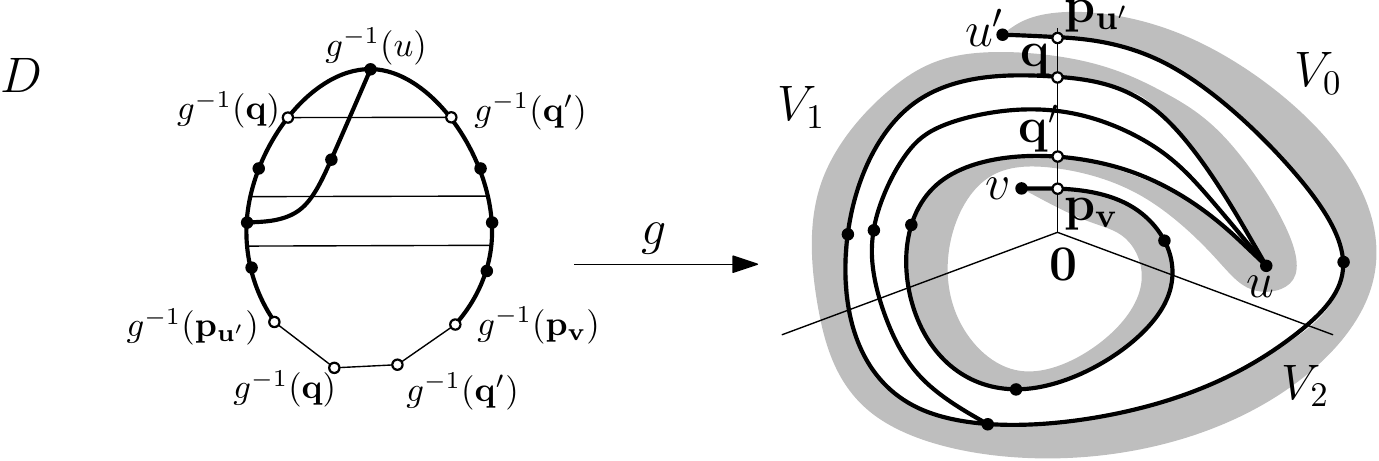}
    	}
    	\hspace{10px}
\subfloat[]{
\label{fig:deform5}
\includegraphics[scale=0.65]{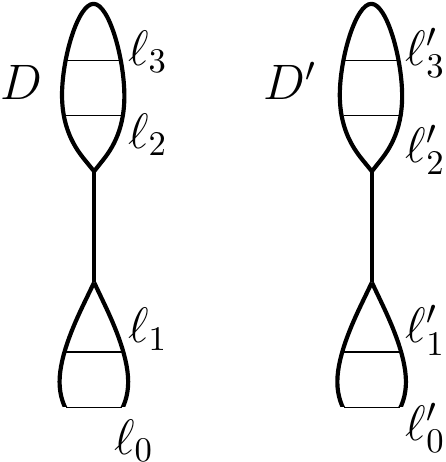}	
    	}
\caption{(a) A ``covering map'' $g$ from a disc onto a spur. (b) Identifying parts of the boundary of the discs $D'$ 
according to $D$.}

\end{figure}

  If $v$ does not have a concave wedge incident to $f$ it could happen that the boundary of $D_0$ crosses itself. As we will see later due to this reason we cannot simply put $D:=D_0$.
We use a ``covering map'' $g$, see Figure~\ref{fig:deform3}.
In the light of Jordan-Sch\"onflies theorem,
  let a topological disc $D$ be a pre-image of a continuous  map  $g:D \rightarrow D_0$ such that the map $g$ is injective when restricted (in the target) to the embedding of $G$;
$g$ maps the boundary of $D$ to the concatenation of ${\bf p_{v}}{\bf p_{u'}}$
with the parts of $P$ and $Q'$ connecting endpoints of ${\bf p_{v}}{\bf p_{u'}}$  with $u$;
and the pre-image (of  parts) of the rays separating clusters
 consists of a union of a connected part of the boundary of $D$ 
 (contained in the pre-image of ${\bf p_{v}}{\bf p_{u'}}$) and
a set of pairwise disjoint diagonals of $D$.
Treating $G$ as a topological space let $G|_D:=g^{-1}(G)$.  

 If $G$ contains cut-vertices $D$ can have pairs of boundary points identified, see Figure~\ref{fig:deform5}.
Let $\ell_0'$ denote  a line segment contained inside the intersection of the interior of $f$ with a ray separating clusters containing ${\bf p_v'}$, whose end vertex is very close to ${\bf p_v'}$, and $\ell_0'\subset{\bf p_v'p_v}$ if and only if $\omega_v$ is convex. Let $D'$ denote a disc bounded by $\ell_0'$ and a curve 
joining the endpoints of $\ell_0'$ following $Q$ towards $v'$ and back in its small neighborhood in the interior of $f$.
Let $\ell_0$ be the connected component on the boundary of $D$ in $g^{-1}({\bf p_vp_{u'}})$.
Let $\ell_i$ for $0<i<|\length(P)|$, denote a connected component of a pre-image by $g$ of (a part of) a ray separating clusters. The segments 
$\ell_i$'s are indexed by the order of appearance of their endpoints along the boundary of $D$.
A line segment $\ell_i$ is just a point if it joins identified boundary points in $D$.
Similarly, let $\ell_i'$ for $0<i<|\length(P)|$, denote a connected component of the 
intersection of $D'$ with rays separating clusters.
We contract $\ell_i'$ to a point iff $\ell_i$ is a single point
and we  contract  interior parts of $D'$ in the correspondence with $D$ and $\ell_i$'s as illustrated in Figure~\ref{fig:deform5}. Since $v'\not\in D_0$, we have $(g(D)=D_0)\cap D' = \emptyset$.

 \begin{figure}[h]
  \centering
\centering
\subfloat[]{
\includegraphics[scale=0.65]{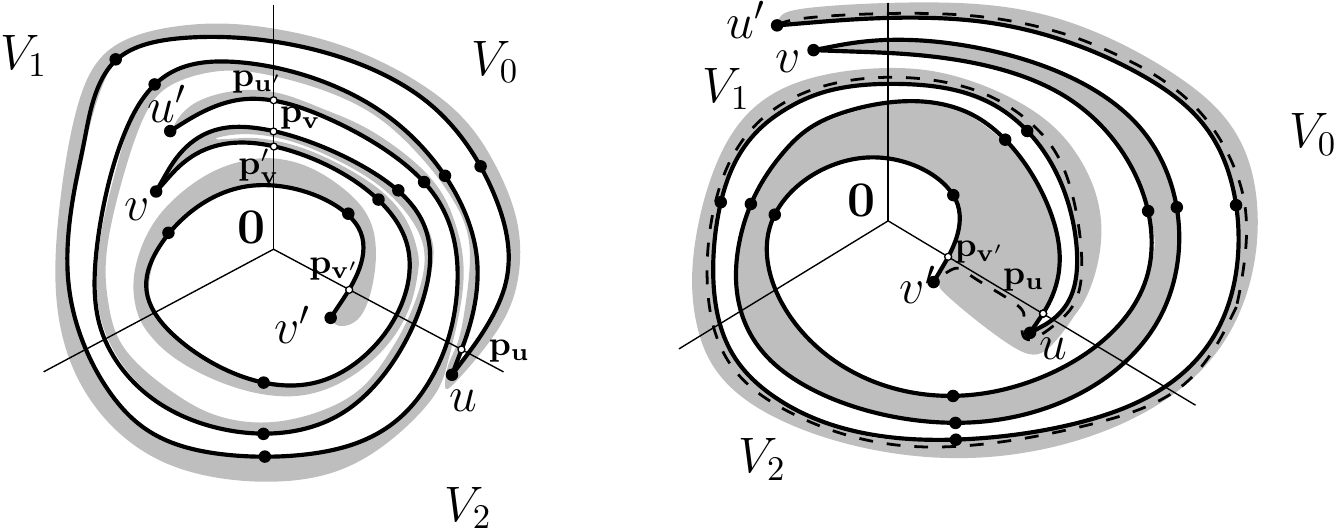}
\label{fig:deform2}    	}
    	\hspace{3px}
\subfloat[]{
\includegraphics[scale=0.65]{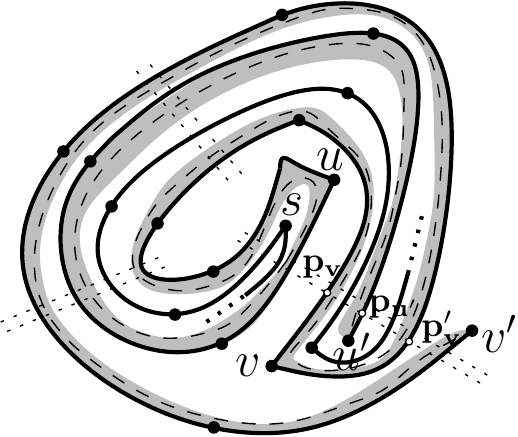}
\label{fig:deform4}    	}

\caption{(a) Deformation in the case when only $u$ has a concave wedge incident to $f$ that is indicated by grey.
The dashed curve represents the desired path subdividing $f$.
On the left, the line segment ${\bf p_{v'}p_{u}}$ crosses  edges of $G$. 
On the right, the corresponding deformation. (b) The wedge at $u$
in $f$, that is indicated by grey, is convex.  A sink $s$  in the interior of a spur having the vertex $u$
as the tip.}
\end{figure}

\begin{figure}
 \centering
\includegraphics[scale=0.7]{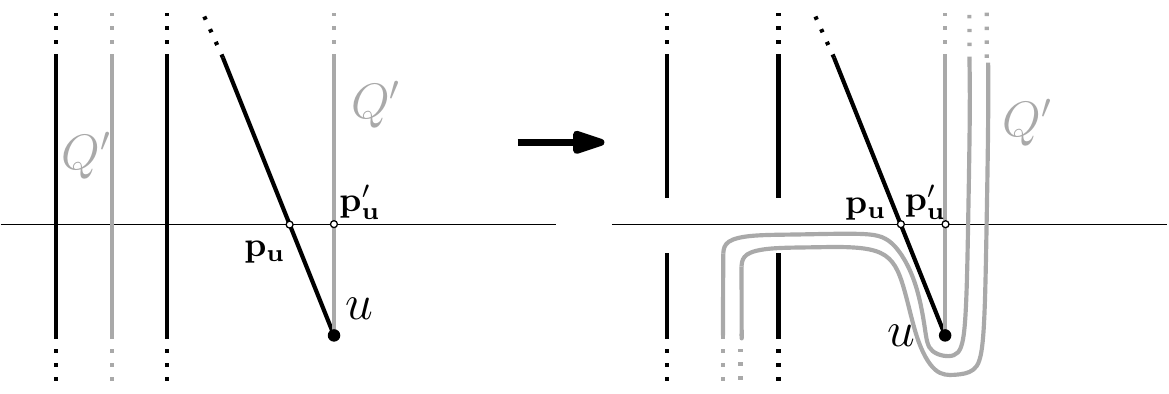}
\caption{Re-routing $Q'$ past $u$ along the boundary of $f$.}
\label{fig:deform8} 
\end{figure}

We map by a homeomorphism $h$ the disc $D$ to $D'$ so that $\ell_i$
is mapped to $\ell_i'$ for all $0\leq i<|\length(P)|$
and so that the endpoint ${\bf x}$ of $\ell_0$ for which $g({\bf x})$ is closer to $\ell_0'$ is mapped to the point of $\ell_0'$ closest to (farthest from) 
$g({\bf x})$ if $\omega_v$ is concave (convex), see Figure~\ref{fig:stork} for an illustration.
We alter the embedding of $G$ by deleting $g(G|_D)$ and replacing it
by $h(G|_D)$.
Finally, we  reconnect the severed end pieces of edges intersecting ${\bf p_vp_{u'}}$ by curves inside the cluster containing $v$ without
creating any edge crossing.

If $v'\in D_0$, we redraw the portion of $G$ contained in the disc $D$ (we override the previous $D$) bounded by the line segment ${\bf p_{v'}p_u}$
and the parts of $P$ and $Q$ joining end points of ${\bf p_{v'}p_u}$  with $v$, see Figure~\ref{fig:deform6}.
Note that the boundary of the disc $D$ is non-self intersecting.
If $Q'$ does not cross ${\bf p_{v'}p_u}$ we make ${\bf p_{v'}p_u}$ crossing free by mapping
the part of $G$ in $D$ to a long skinny disc $D'$  (we override the previous $D'$) in the vicinity of $Q'$ and reconnecting the severed edges similarly 
as above.

 \begin{figure}
  \centering
\centering
\subfloat[]{
\includegraphics[scale=1]{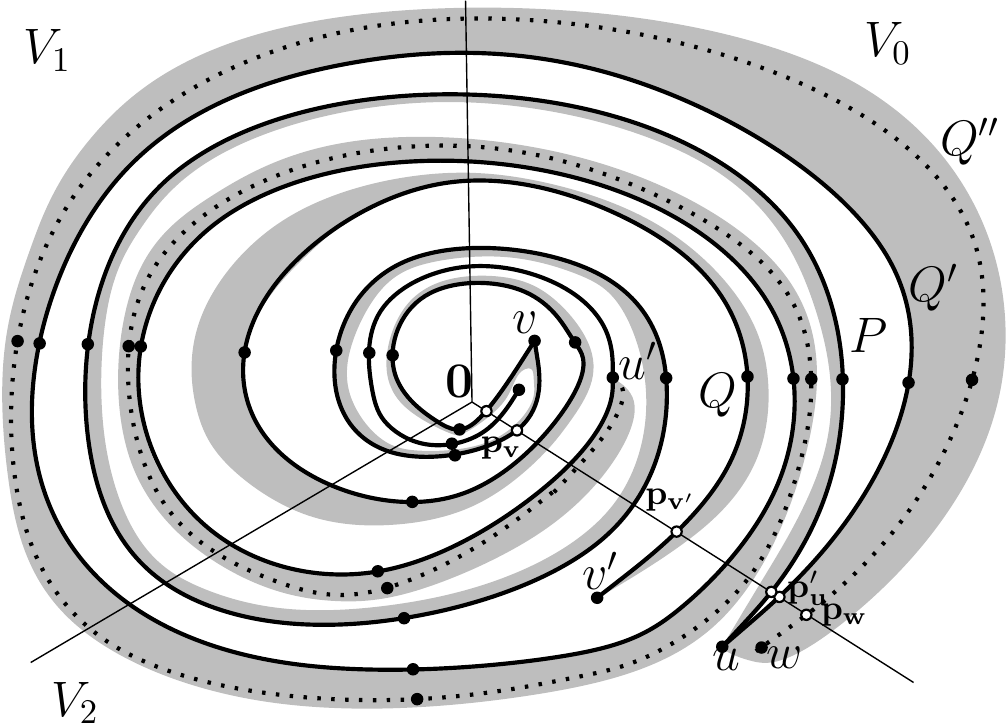}
\label{fig:deform6}  }
    	\hspace{1px}
\subfloat[]{
\includegraphics[scale=1]{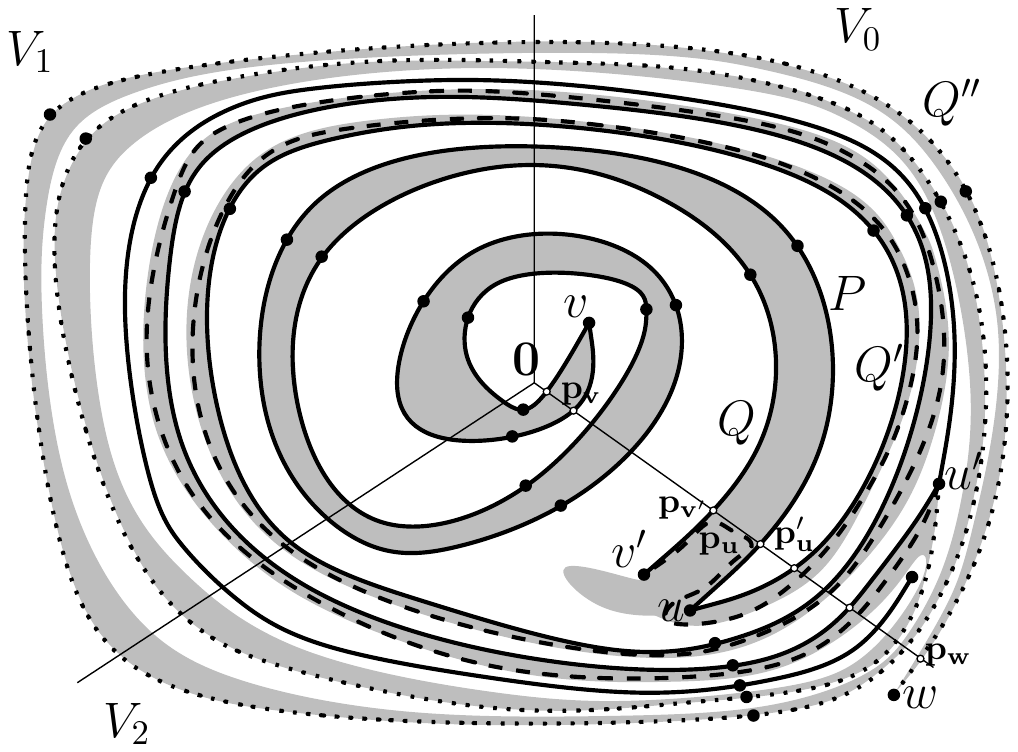}
\label{fig:deform7}}
\caption{The face $f$ is indicated by grey. (a) The vertex  $v'\in D_0$ corresponding to the spur with the tip $u$. (b) Corresponding  deformation involving $Q'$ and the dashed path $P'$ subdividing $f$.}
\end{figure}

If $Q'$ crosses ${\bf p_{v'}p_u}$ we cannot apply the previous argument since $D \cap D' \not=\emptyset$,
and thus, we need a different approach.
We temporarily add to $G$ a path $u'Q''w$ starting at $u'$, following closely $Q'$ in the interior of $f$
and ending in a vertex $w$. Let ${\bf p_w}$ denote the crossing point of the last edge on $Q''$
separating clusters with a ray separating clusters. 
We cut edges at ${\bf p_{v'}p_u}$ by removing a small $\epsilon>0$ neighborhood of their crossing points.
We reroute the paths $Q'$  (see Figure~\ref{fig:deform8}) and $Q''$ without crossing an edge of $G$ from their severed ends outside of $D$ past $u$ and along $Q'$ in the interior of $f$ so that we closely follow the boundary of $f$.
We  easily avoid creating crossings since we cut edges at ${\bf p_{v'}p_u}$.
Let $D''$  be the disc bounded by rerouted parts of $Q'$ and $Q''$ from their first intersections on the way from $u$ and $w$, respectively, with ${\bf p_wp_{u}'}$ to their common vertex $u'$; and by $\ell_0''\subset {\bf p_wp_{u}'}$ connecting the ends of those parts.
Let $\ell_i''$, $0\leq i < |\length(P)|-c$, denote the line segments in the intersection of $D''$ with rays separating clusters listed in the order of appearance along the boundary of $D''$.
We deform the embedding of $G$ by a mapping $g$ from $D\setminus (P \cup Q)$, see Figure~\ref{fig:deform7} for an illustration, such that $g$ maps the relative interior of $\ell_0={\bf p_{v'}p_u}$ to the relative interior of  ${\bf p_wp_{u}'}$, the parts of $Q'$ and $Q''$ in $D$ to their rerouted counterparts, and for the open line segments $\ell_i$, $1\leq i <|\length(P)|-c$, in the intersection of $D$ with the rays separating clusters we have $g(\ell_{i+c})\subset \ell_i'' \subset g(\ell_i)$.
The mapping $g$ is then extended to the whole $D\setminus (P \cup Q)$ such that (i) $g$ is a homeomorphism when restricted to the interior
of the slab between $\ell_i$ and $\ell_{i+1}$, for all $0\le i<|\length(P)|$, where $\ell_{|\length(P)|}$ is the 
final piece of the boundary of $D''$, and (ii) no edge crossings are introduced, i.e., $g(D|_{G})$ is injective
and $g(D|_{G}) \cap G =\emptyset$. It is not hard to see that $g$ exists. Finally, we reconnect the severed end pieces of edges inside the cluster
containing $u$ and remove $Q''$.
The previous deformation is perhaps easier seen as follows. The clustered embedding of $G|_{D}$ can be made  arbitrarily skinny. Thus, for the purpose of  deformation, we can picture that $G|_{D}$ consists just of the part of $Q'$ and $Q''$ in $D$ and a strictly monotone 
path starting at  $u'$ as in Figure\ref{fig:deform6}. We map $G|_{D}$ as indicated in Figure~\ref{fig:deform7}.

Second, if both $\omega_u$ and $\omega_v$ are convex we can subdivide $f$ by $P'$ unless ${\bf p_{v}p_{u'}}$ is intersected by edge(s) of $G$ (we still assume that  ${\bf p_v'}\not\in{\bf p_{v}p_{u'}}$).
   However, if this is the case $D_0$ (defined above) contains a sink or source $s$ in its interior, see Figure~\ref{fig:deform4}. Note that we can assume that $s$ is also a tip of a shortest spur of $f$, and hence, the previous case applies.
\qed\end{proof}

\end{document}